\documentclass[runningheads]{llncs}

\usepackage[misc]{ifsym}

\usepackage{graphicx}
\usepackage{mathrsfs}
\usepackage[title]{appendix}
\usepackage[linesnumbered,ruled,vlined]{algorithm2e}
\SetAlFnt{\small}
\SetAlCapFnt{\small}
\SetAlCapNameFnt{\small}
\usepackage{algorithmic}
\algsetup{linenosize=\tiny}

\SetKwProg{Fn}{Function}{}{end}\SetKwFunction{FRecurs}{FnRecursive}%
\SetAlgoLongEnd
\DontPrintSemicolon

\begin{document}
\title{Multi-MEC Cooperation Based VR Video Transmission and Cache \\ using K-Shortest Paths Optimization\thanks{Supported by the Sichuan Provincial Science and Technology Department Project (No.2021YFHO171).}}

%
\titlerunning{Multi-MEC Cooperation Based Transmission and Cache for VR video}
%
\author{Jingwen Xia \and Luyao Chen \and Yong Tang\textsuperscript{(\Letter)} \and Ting Yang \and Wenyong Wang} 

\institute{University of Electronic Science and Technology of China, Chengdu, China \email{worldgulit@uestc.edu.cn}}

\authorrunning{J. Xia et al.}
%
%
\maketitle              

\begin{abstract}
In recent network architectures, multi-MEC cooperative caching has been introduced to reduce the transmission latency of VR videos, in which MEC servers' computing and caching capability are utilized to optimize the transmission process. However, many solutions that use the computing capability of MEC servers ignore the additional arithmetic power consumed by the codec process, thus making them infeasible. Besides, the minimum cache unit is usually the entire VR video, which makes caching inefficient.

To address these challenges, we split VR videos into tile files for caching based on the current popular network architecture and provide a reliable transmission mechanism and an effective caching strategy. 
Since the number of different tile files $N$ is too large, the current cooperative caching algorithms do not cope with such large-scale input data. 
We further analyze the problem and propose an optimized k-shortest paths (OKSP) algorithm with an upper bound time complexity of $O((K \cdot M + N) \cdot M \cdot \log N))$, and suitable for shortest paths with restricted number of edges, where $K$ is the total number of tiles that all $M$ MEC servers can cache in the collaboration domain. And we prove the OKSP algorithm can compute the caching scheme with the lowest average latency in any case, which means the solution given is the exact solution. 
The simulation results show that the OKSP algorithm has excellent speed for solving large-scale data and consistently outperforms other caching algorithms in the experiments.

\end{abstract}
\section{Introduction}
Virtual Reality (VR) is emerging as a killer use case of 5G wireless networks. Users can watch VR videos with mobile devices such as VR glasses to obtain an immersive experience. The VR video mentioned in this paper refers to the panoramic video of horizontal 360$^\circ \times$ vertical 180$^\circ$.

The transmission problem of VR video has attracted extensive attention in academia and industry \cite{hsu2020mec}. Under the premise of ensuring a good viewing experience, the bit rate required for VR video is much higher than that of traditional video. If users fetch all the VR videos they need from the data center, the bandwidth consumption of the 5G network backhaul link will increase dramatically. Besides, an immersive VR experience requires the network to provide ultra-low transmission latency (less than 20 ms) to avoid user discomfort \cite{sun2019communications}. However, under the traditional network architecture, the traffic in the core network and the Internet changes dynamically, which will cause fluctuations in the network load, resulting in unpredictable performance degradation of watching VR videos. 

The mobile edge computing (MEC) proposed in the 5G network technology can greatly alleviate the above problem \cite{jedari2020video}. The MEC server is introduced in the traditional network architecture and deployed near the base station where the user accesses. The MEC server caches the current popular videos, thus significantly mitigating the request latency of mobile VR devices. However, compared with data centers, the resources of a single MEC server is very insufficient \cite{ndikumana2019joint}. To fully utilize the limited storage space to meet more user requests, it is an inevitable trend to collaborate with multiple MEC servers in close geographical proximity and share their storage space.

The computing capability of the MEC server can also be utilized for VR video transmission. Before further analysis, we need to know the process of VR video from production to rendering \cite{mangiante2017vr}, which includes (1) \textit{shooting and stitching}: capturing video by multi-camera array, then projecting the captured video frames into the spherical space and then stitching to get the panoramic video; (2) \textit{mapping and encoding}: mapping and projecting the panoramic video image into the 2D plane (e.g., using equirectangular projection technique) to get the 2D plane image, and then encode the video using universal video encoder; (3) \textit{transmission}: transmitting the encoded video to the VR device through the network; (4) \textit{rendering}: the VR device decodes the VR video, and projects the 2D video frames into the spherical space and then displays them on the screen.

Many papers \cite{cheng2021design,hsu2020mec,liu2021rendering,sun2019communications} have studied offloading the rendering task in (4) to the MEC server, and later transmitting the rendered video to the VR device, to reduce the overall latency. However, there are many problems with this optimization. 

The operation involved in step (4) is simply applying a series of trigonometric functions to the coordinates of each point in the 2D video frame. Constant time complexity can be considered for an operation on a particular pixel point \cite{moon1988spherical}, so the rendering process does not take much time. After offloading the rendering task to the MEC server, the VR video needs to be decoded before rendering. The decoding operation will take extra time. Because the MEC server may be equipped with a GPU dedicated to processing video, it will reduce the rendering time. However, after rendering, the VR video size equals the original uncompressed video size, which depends on the compression ratio (could be as high as 1000:1 in H.265 encoding \cite{noauthor_meeting_2013}). Without compression, transmitting a video with such a high bitrate is impractical. But if the MEC server compresses the rendered video, the time consumed by additional codecs is even higher than the rendering time saved, leaving aside the loss of picture quality. Overall, using the computing capability of the MEC server to handle the rendering task does not improve the user viewing VR video experience but makes the problem more and more complicated.

But the computational capability of the MEC server can be utilized for compute-intensive tasks such as predicting user head motion to reduce the amount of data transfer. When a user watches a 360$^\circ$ VR video, the field of view (FoV) usually accounts for only about 12.5\% of the panoramic field of view (90$^\circ \times $ 90$^\circ / $ 360 $^\circ \times $ 180$^\circ$ \cite{hou_predictive_2021}). In other words, if the entire VR video is transmitted, most of the frame area is not visible to the user during viewing. To reduce the problem of large bandwidth consumption and high latency caused by the transmission of 360$^\circ$ surround images, most existing transmission schemes \cite{guntur2012tile,kimata2012mobile,mavlankar2010video} transmit only the high bit rate of the image within the user's FoV. The specific operation is that in the VR video production process (2), the panoramic video after equirectangular projection is spatially and equidistantly divided into $x \times y$ independently codable and decodable whole-length tile video files \cite{cheng2021design}, whose duration is equal to that of the original video. Due to the users' gaze shifting while watching the video, the request rate of whole-length tile video fluctuates with time and maybe low at certain times \cite{mahzari_fov-aware_2018}. If the whole-length tile video file is the minimum cache unit, the MEC server's cache efficiency is low, so it can be split into equal-sized tile videos in time, which we call tile files or tiles.

If the request only occurs when the VR device detects the need to render a particular tile file, it will result in a very short tolerant latency (e.g., less than 20 ms), and the user experience will be poor. The split operation will be useless if the whole frame is cached. Thus, it is necessary to predict and cache the tile files that users may watch in advance, which requires the prediction of the users' head motion trajectory in the coming period. Related research \cite{rondon_unified_2020} shows that various head motion prediction models have been proposed to predict short-term head trajectory changes with high accuracy in recent years. In the tile-based prediction model presented in the literature \cite{fan_fixation_2017}, prediction accuracy can reach 84.22\% when the prediction period is set to 1s.

Even though the existing prediction model has high accuracy, there are still cases in which the required video frame is not requested in advance when the VR device is ready to render. In this case, if the entire tile file belonging to the missing video frame is requested, it is easy to time out and causes the user to feel the local screen lag. To solve this problem, servers can extract the currently missing tile frames for transmission instead of transmitting the entire tile files. However, only transmitting the missing frames presents a problem. Because in the mainstream video encoding \cite{filippov_video_2020} includes 3 types of frames: I (Intra-coded), P (Predictive), B (Bi-predictive). I frames include complete image data. P or B frames contain information about the changes between the previous and subsequent I or P frames (which is used to create the resulting image). So if the missing frame is a P or B frame, it cannot be directly decoded on the VR device. As a result, a feasible remediation method on the server-side is to reconstruct the missing frame, compress it into an I frame, and then respond to the VR device.

In summary, head motion prediction and remediation to deal with prediction failure can be offloaded to the MEC server, allowing the VR device to focus on video decoding and displaying. Based on the optimization points analyzed above, we aim to build an efficient mobile VR video transmission framework with a multi-MEC cooperative caching mechanism. Our contributions are summarized as follows:
\begin{itemize}
	\item[-] We describe a multi-MEC cooperative network architecture and provide a reliable transmission mechanism.
	\item[-] For the large-scale tile caching problem, we analyze the caching profit of tiles under the cooperative architecture and equivalently rebuild the tile caching problem into the k-shortest problem.
	\item[-] Due to the high time complexity of the original k-shortest paths algorithm \cite{suurballe_disjoint_1974} in this paper's application scenario, we propose an optimized k-shortest paths (OKSP) algorithm. And upper bound time complexity of the OKSP algorithm is $O((K\cdot M + N) \cdot M \cdot log N))$, where $K$ is the total number of tiles that all $M$ MEC servers can cache in the collaboration domain, and $N$ is the number of different tile files. We also prove the solution of the OKSP algorithm is the exact solution.
	
\end{itemize} 

The remainder of this paper is organized as follows.  Sect.\ref{sec:Related Work} reviews the relevant work. Sect.\ref{sec:System Architecture} presents the system architecture and transmission mechanism. Sect.\ref{sec:Cache Strategy} analyzes the cache profit and shows the specific optimization details of the OKSP algorithm. Sect.\ref{sec:Experiment} discusses the simulation results. Sect.\ref{sec:Conclusion} concludes this paper.

\section{Related Work}
\label{sec:Related Work}
In recent years, some papers have utilized the characteristics of VR video in their transmission network architecture, since the user can only view a portion of the panoramic video at any given moment, transmitting the user's current FoV results in significant savings in bandwidth consumption. As a result, \cite{gaddam_tiling_2016} proposes a tile-based scheme to transmit VR videos. \cite{fan_fixation_2017} further improves the quality of user experience by deriving a user's head motion trajectory model based on real data training. \cite{fan_fixation_2017,gaddam_tiling_2016} focus on the prediction of head motion when users watch VR videos, describing the details of the prediction model and giving experimental results.
The linear regression method is adopted for trajectory tracking in \cite{hu2019optimal,nguyen2019optimal}. The authors in \cite{ndikumana2019joint} propose a method for determining tile resolution based on user head movements and network conditions. 
A 360$^\circ$ video streaming scheme is proposed in \cite{hu2019optimal} to maximize user experience by predicting the user's viewport and prefetch video segments that will be viewed soon.

Some studies adopt different bitrate transmission for the inner and outer regions of FoV to reduce the data transmission volume. Ghosh et al. \cite{ghosh2017rate} assume the FoV is known exactly beforehand before fetching the invisible portion at the lowest resolution and determining the visible portion's resolution using bandwidth prediction. A heuristic method \cite{ozcinar2017viewport} is proposed for allocating video resolution by assigning different weights to different regions inside and outside the current FoV. \cite{ozcinar2018omnidirectional,zhang2018cache} calculate the popularity of the current FoV to predict the popularity of FoV for a while in the future, among which \cite{ozcinar2018omnidirectional} proposes a novel algorithm for assigning bitrate based on viewpoints.

Also, some papers use the aid of MEC technology in VR video transmission. \cite{cheng2021design,sun2019communications} propose a MEC-based network architecture to deliver mobile VR videos. In \cite{sun2019communications}, a FoV-based caching architecture is proposed to save bandwidth consumption and meet the low latency requirement of VR videos. \cite{cheng2021design} proposes a MEC-based architecture for 360$^\circ$ VR video delivery and actively exploits the predictable nature of user FoV to cache VR videos. However, both studies work only with a single MEC server and do not consider the potential of multi-MEC cooperation.

Hou et al. \cite{hou_proactive_2018}, Ndikumana et al. \cite{ndikumana2019joint} and Wang et al. \cite{wang_zone-based_2019} notice the advantages of multi-MEC cooperation. Wang et al. \cite{wang_zone-based_2019} propose a zone-based cooperative content caching and delivery scheme and develops a heuristic cooperative content caching strategy. Hou et al. \cite{hou_proactive_2018} also propose a service architecture based on multi-MEC cooperation and a cooperative caching strategy based on migration learning under the service architecture. Ndikumana et al. \cite{ndikumana2019joint} propose a joint computation, caching, communication, and control optimized caching strategy based on a multi-MEC cooperative architecture. However, these studies all concentrate on general video content and do not exploit the features of VR video.

Using multi-MEC cooperative transmission, some studies exploit VR's characteristics. Ge et al. \cite{ge_multipath_2017-1} propose a multipath cooperative routing scheme to facilitate transmission of 360-degree MVRV among edge data centers, base stations, and users. A mobile edge computing system \cite{yang_communication-constrained_2018} is proposed which utilizes the computing and caching resources of mobile VR devices to reduce communication-resource consumption. To enable high-quality VR streaming, Chakareski et al. \cite{chakareski_vrar_2017} integrate scalable layered 360 video tiling, viewport-adaptive rate allocation, optimal resource allocation, and edge caching.

\section{System Architecture}
\label{sec:System Architecture}

\subsection{Network Architecture}
\begin{figure}
	\centering
	\includegraphics[width=0.8\textwidth]{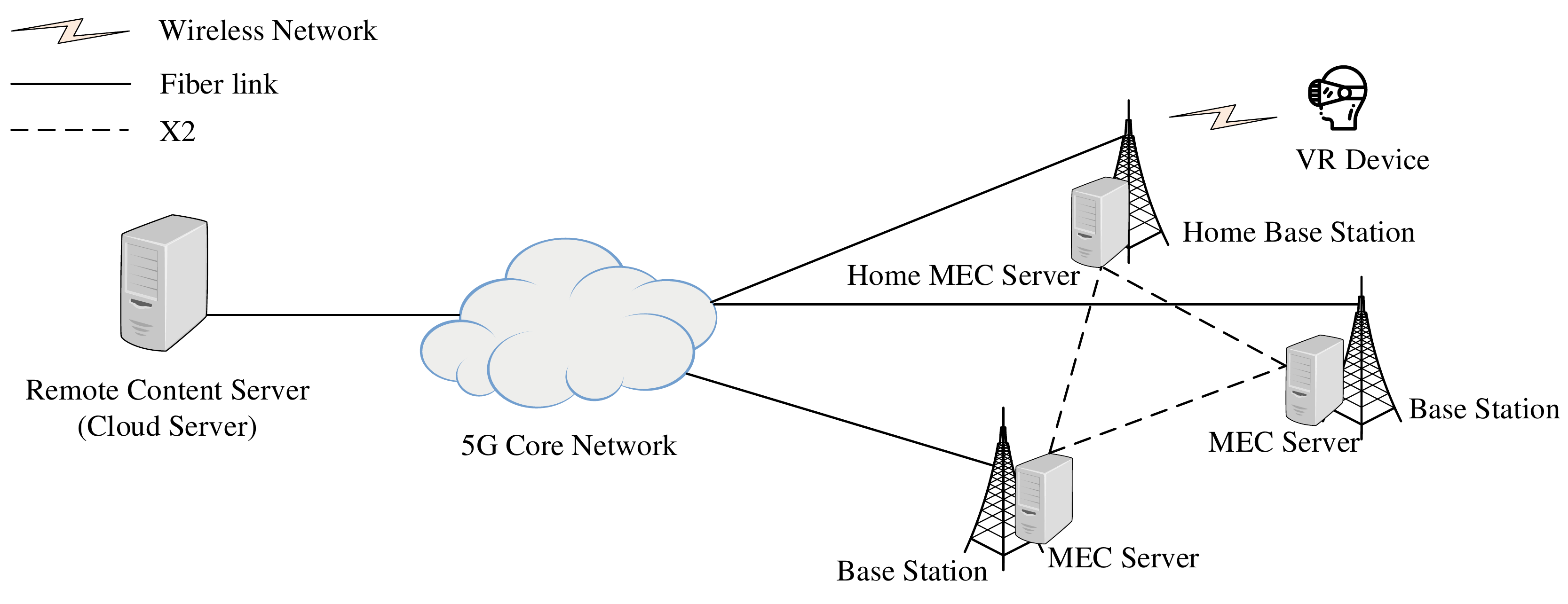}
	\caption{Network architecture diagram.} \label{net-arch-fig}
\end{figure}
The network architecture based on multi-MEC cooperation is shown in Fig.~\ref{net-arch-fig}. The network architecture consists of three parts: a remote content server(e.g., cloud server) located in a data center, a 5G core network, and a MEC collaboration domain. We assume the cloud server has a massive storage space on which all the encoded tile files are stored. Since the coverage area of a base station is relatively tiny ($\le$ 500m \cite{noauthor_how_2020}), 
multiple base stations are usually deployed in a living/commercial area, so these base stations are close geographically. A MEC collaboration domain can be artificially created by dividing MEC servers near these base stations in order to share computing and caching resources.

Assume that there are $M$ MEC servers in a collaboration domain, which is denoted as $\{$MEC server 1,..., MEC server $M\}$. Each MEC server maintains a resource access table, which records the tile files it cached. The base station directly connected to the VR device are called home base station and the MEC server deployed near home base station are called home MEC server.


\subsection{Transmission Process}
As discussed in previous, it would be necessary for VR devices to prefetch the tiles that will be used in the future period. Various head motion prediction models have recently been proposed to predict short-term trajectory changes with high accuracy \cite{rondon_unified_2020}. Therefore, the prediction task must continuously predict during the user's viewing process. The prediction work requires a certain amount of computation, and offloading it to the home MEC server can relieve the load on the processor of the VR device. While the VR device focuses only on decoding and displaying VR videos, the home MEC server is responsible for continuously predicting the user's short-term head motion trajectory and transmitting the tile files that may be involved based on the predictions.

However, due to the prediction bias, the VR device may find that the tile file is not requested in advance when rendering the VR video. In this case, the VR device must immediately request the currently missing tile video frames. The user's current viewing information can be sent to the home MEC server each time the VR device refreshes the screen to detect the prediction misses early and remedy the situation.

In this paper, the VR device and the home MEC server each perform the following transmission process.

\textbf{For VR Devices.} When the user starts watching the video, the VR device establishes a long connection with its home MEC server. The current playback status is detected every screen rendering interval $\Delta t_{fps}$ time. If the video playback is finished, the interaction process end and the VR device disconnects from its home MEC server. Otherwise, perform the screen rendering operation; simultaneously, the user's current FoV and the video playback time $t_v$ are sent to its home MEC server.

\textbf{For MEC servers.}The home MEC server will record the cache status of the users within its service area, when the MEC server receives an information packet from a VR device, it can judge whether the user is missing the render frames based on the FoV and  $t_v$ of this packet. Perform the remediation operation first if it is missing; otherwise, perform the prediction operation directly.

\textbf{Remediation Operation.} 
As mentioned in the introduction, since the request's tolerance latency (20ms) is rather harsh at this point, the home MEC server responds by returning the missing frames of the VR device instead of the tile files.
Based on the resource access table, the server will find all the generated missing video frames in this rendering process and sends them back to the VR device, and then perform prediction operation.

\textbf{Prediction Operation.} The home MEC server adds the user's current FoV to the user's trajectory and runs the prediction model to predict the user's trajectory for the next period $(t,t+\Delta t]$ and the tile files involved in this trajectory.
After that, the server will find all the tile files that the device is not cached and sends them back to the VR device.
Since a tile file contains multiple frames, some may be earlier than $t_v$, which means they will never be watched, so unless necessary, we only send the frames later than $t_v$. 

The exceptions are due to the video compression method. P and B frames cannot be directly decoded on the VR device, so we need to look forward to the nearest I frame position and send from that position. Considering the size of a single tile file is small, all frames later than $t_v$ are sent together even if they may not be rendered in the future.
	
\section{Cache Strategy}
\label{sec:Cache Strategy}
To design an ideal caching algorithm, we will analyze the latency involved in each possible request path, establish a profit model for the amount of latency reduction, and finally propose an algorithm for the k edge-disjoint shortest paths problem with a constraint on the edge number of shortest paths. The algorithm will eventually give the highest profit caching scheme.

\subsection{Delay Analysis}
When computing the caching scheme, assume that $N$ different tiles have been requested in the entire MEC collaboration domain, and the tiles are denoted as $\{$tile$_1$,...,tile$_N\}$. A request$_r$ about tile$_n$ from a VR device will have a data size of $sz_r$, a bandwidth of $Bu_r$ between the user and MEC server, a bandwidth of $B_m$ between MEC servers, and a bandwidth of $B_{dc}$ between MEC servers and the cloud server. Both request and response packets on the MEC servers will experience queuing before they are processed and sent. We assume that all MEC servers in the collaboration domain are moderately loaded when running, so the queuing delay would not be too long. The $T_{q_1}$ is used to represent the upper limit of delay time of waiting to be processed after the request arrives at the receive buffer of the MEC server, and $T_{q_2}$ is used to represent the upper limit of queuing delay before the response data is sent. All request goes through one or two communication processes as follows.\\

\textbf{VR Device $\leftrightarrow$ home MEC server.} There is only one hop between the VR device and its home MEC server. Since radio waves travel at the speed of light, and the distance between them is usually less than 1 km, propagation delays can be negligible. The delay of this case is
\begin{equation}
	ch(sz_r,Bu_r) = \frac{sz_r}{Bu_r} + T_{q1} + T_{q2} .\\
\end{equation}

\textbf{Home MEC server $\leftrightarrow$ Other MEC server.} Base stations in the collaborative domain communicate via the X2 interface\cite{ndikumana2019joint}, and the underlying physical link is optical fiber, a one-hop wired link. It is not more than 10 km between the MEC servers, so propagation delay can also be negligible. MEC servers in the collaborative domain have similar network configurations, so the communication bandwidth between any two servers is the same. The delay of this case is
\begin{equation}
	cmm(sz_r) = \frac{sz_r}{B_m} + T_{q1} + T_{q2} .\\
\end{equation}

\textbf{Any MEC server $\leftrightarrow$ Data  Center.} As the network conditions from the home MEC server to the remote data center are complex, including multi-hop links and network fluctuations due to line instability, it is impossible to calculate the delay accurately. In order to determine the time it takes packets to pass through these links, we use the communication delay $ T_{m \leftrightarrow dc}$, which fluctuates within a range. The delay of this case is
\begin{equation}
	cc(sz_r) = \frac{sz_r}{B_{dc}}+2 \cdot T_{m \leftrightarrow dc}.
\end{equation}

As shown above, the system of equations can be used to estimate the time delay of both prefetch and remediation requests. However, remediation requests will incur processing latency for extracting the frames required for compression from the missing tile file. Because the closest I frame information is on-demand, remediation can only be performed on the MEC server with the missing tile file. The processing delay is independent of the cache strategy. As a result, we do not consider that the processing delay can be reduced.

Afterwards, we discuss how the cache distribution of each tile file affects the request delay. Let the requests about tile$_n$ on MEC server $m$ be $R_{m,n}$. We find that once any replica of tile$_n$ is in the collaboration domain, it generates a global profit
\begin{equation}
	v_n^{glo} = \sum_{m = 1}^{M}\sum_{r \in R_{n, m}} cc(sz_r) - cmm(sz_r),
\end{equation}
Because once a replica exists within the collaboration domain, the path of all requests to tile$_n$ shifts and the request endpoint changes from the cloud server to the MEC server within the collaboration domain. This global profit can only be accrued once. Also, The tile$_n$ cache on a special MEC server $m$ generates a local profit
\begin{equation}
	v_{n,m}^{loc} = \sum_{r \in R_{n, m}}cmm(sz_r),
\end{equation}
Since all tile$_n$ requests within the service area of MEC server $m$ are handled directly, inter-MEC server communications are eliminated.

\subsection{Abstraction Method}
\begin{figure}
	\centering
	\includegraphics[width=0.7\textwidth]{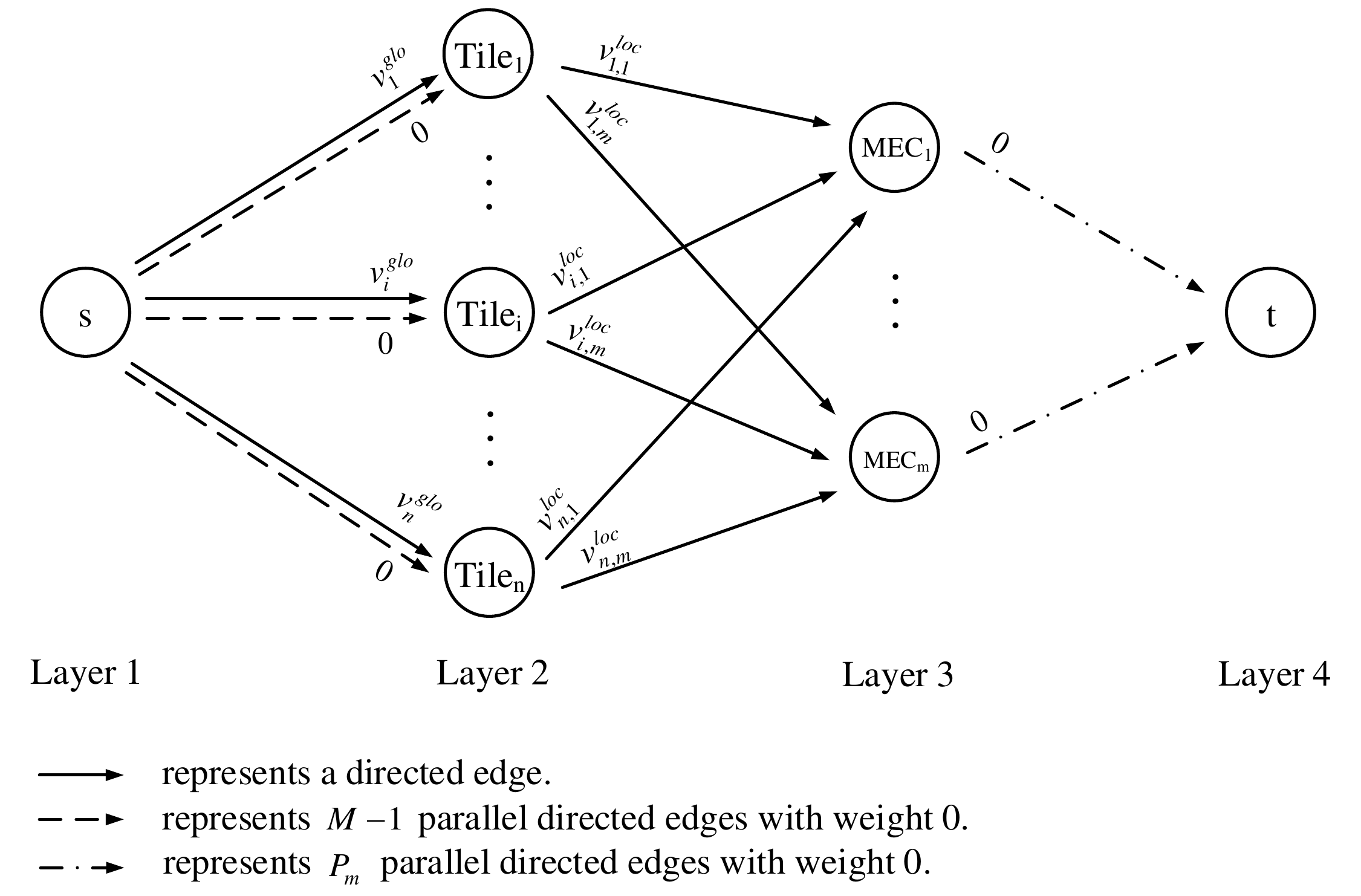}
	\caption{The illustration on reformulating cache problem to k-shortest paths problem.} \label{ksp1-fig}
\end{figure}
After we split all original VR videos into equal size tile files, Let $S$ denote the size of a tile file, and $P_m=\frac{C_m}{S}$ denote the maximum number of tile files stored in MEC server $m$, where $C_m$ is the storage capacity of the MEC server $m$. In our specific case, the cache distributed problem can be reformulated as the graph structure $G=(V, E)$ with a single source and a single sink illustrated in Fig.~\ref{ksp1-fig}, which is a directed acyclic graph (DAG). 

Nodes with depth $d$ are called layer-$d$ nodes for convenience. In the DAG $G$, there are only four layers. Layer-1 is the source node $s$, layer-2 is the tile nodes corresponding to each tile file, layer-3 is the MEC nodes corresponding to the MEC servers in the collaboration domain, and layer-4 is the sink node $t$.

If there is only one directed edge from node $i$ to node $j$, it will be represented as $e_{i, j}^t$, and if there is more than one, $e_{i, j}^t$ will be used to represent the $t^{th}$ parallel directed edges. Edge weights are expressed as $c(e_{i, j}^{t})$. The details of the $G$ will be described below.

\textbf{\boldmath{$s$} $\rightarrow $ Tile Nodes.} In the collaboration domain, each tile file will store at most $M$ replicas. Therefore, there are $M$ parallel directed edges \{$e_{s, i}^{1},...,e_{s, i}^M$\} from node $s$ to tile node $i$. Since the global profit can only be accumulated once, $c(e_{s, i}^{1})$ = $v_i^{glo}$, and $c(e_{s, i}^{2}),...,c(e_{s, i}^M)$ should be 0. 

\textbf{Tile Nodes $\rightarrow $ MEC Nodes.} For any layer-3 MEC node $j$, its optional storage object includes all tile files, so for any layer-2 tile node $i$, there is a directed edge $e_{i, j}$, and $c(e_{i, j}) = v_{i,j}^{loc}$. 

\textbf{MEC Nodes $\rightarrow$ \boldmath{$t$}.} To limit MEC's storage capacity, for any layer-3 MEC node $j$, there are $M$ parallel directed edges \{$e_{j, t}^{1},...,e_{j, t}^M$\} from $j$ to $t$ with 0 weight.

Any path between source node $s$ and sink node $t$ in $G$ represents a caching operation in the original problem, selecting a simple directed path  $<s, i, j, t>$ is equivalent to storing tile file $i$ on MEC server $j$. After negating the weights of all edges, the original problem is reformulated as a \textit{k shortest edge-disjoint paths} problem, which has been extensively studied in the literature on network optimization \cite{berclaz_multiple_2011}.

Let $K = \sum_{m=1}^{M} P_m$, which indicates the maximum number of tile files the collaboration domain can store. Our goal became to find the $K$ paths \{$p_1, ..., p_K$\} from source node $s$ to sink node $t$, such that the total cost of the paths is minimum, and the problem can be solved using the k-shortest paths (KSP) algorithm \cite{suurballe_disjoint_1974}. 

Algorithm \ref{orignal-skp} gives a pseudocode summary of the original KSP algorithm, and a brief workflow of KSP algorithm can be described as (1) construct the graph $G$ as mentioned above; (2) compute the shortest path $p$ from $s$ to $t$; (3) in $G$, reverse the direction of all edges in $p$, and negating their weights; (4) repeat step (2)(3) $k-1$ times.

In $l^{th}$ iteration, an augmenting path $ap_l$ with be calculated by shortest path algorithm, let $AP_l = \{ap_1, ..., ap_l\}$ be the set of all $l$ augmenting paths computed up to iteration $l$, and the negative edges of $G$ will form the shortest paths set $P_l = \{p_1, ..., p_l\}$. In particular, after every iteration, there are no negative cycles in the graph $G$, which is proved in \cite{suurballe_disjoint_1974}.

\subsection{Model Analysis}

\begin{table}[ht]
	\centering
	\caption{Notation Definition}\label{tab1}
	\begin{tabular}{|c|c|}
		\hline
		\textbf{Symbol} & \textbf{Description} \\
		\hline
		$N$ & Number of different tile files  \\
		\hline
		$M$ & Number of MEC servers in the collaboration domain \\
		\hline
		$P_m$ & The MEC server $m$ can store a maximum of $P_m$ tile files \\
		\hline
		$K$ & The collaboration domain can store a maximum of $K$ tile files \\
		\hline
		$e_{i, j}^{t}$ & The $t^{th}$ parallel directed edge from node $i$ to node $j$ \\ 
		\hline
		$c(e_{i, j}^{t})$ & The weight of edge $e_{i, j}^{t}$\\ 
		\hline
	\end{tabular}
\end{table}
The worst-case complexity of Algorithm \ref{orignal-skp} is $O(K \cdot M \cdot N)$, which means the original KSP algorithm can not handle the situation where $N$ becomes larger. Below, we demonstrate how some optimizations can significantly reduce the algorithm's time complexity. We find that the layer-3 nodes are significantly fewer than layer-2 nodes ($M \ll N$), and we can optimize the KSP algorithm from this point. Eventually,  we propose the \textit{optimized k-shortest paths} (OSKP) algorithm.

\begin{theorem}
	In $G$, there exists a shortest path set $AP^\ast_K = \{ap_1,...,ap_K\}$; for any path $ap \in AP^\ast_K$, every node in $ap$ will only appear once, except for the tile nodes in layer-2.
\end{theorem}
\begin{proof}
	Since there are no negative cycles in $G$ at each iteration, and removing non-negative cycles does not increase the path length, we can run the KSP algorithm to get a set $AP_K$, then remove cycles in the paths that do not satisfy the requirement repeatedly until all paths in $AP_K$ are meet the requirements.
\end{proof}

\begin{theorem}\label{Theorem 2}
	In the augmenting pathfinding process, if the path's endpoint is at layer-3 node $j$, then the condition to transfer via layer-2 node $i$ to another layer-3 node $k$ is that MEC server $j$ stores tile$_i$ and MEC server $k$ does not.
\end{theorem}
We refer to the process in Theorem \ref{Theorem 2} as a \textit{transfer}, it results in a \textit{transfer loss} of $v_{i,j}^{loc} - v_{i,k}^{loc}$, and we denote $i$ as a \textit{transit node}.

\begin{theorem}
	For any $ap \in AP^\ast_K$, the edge number of $ap$ does not exceed $2 \cdot M + 1$.
\end{theorem}
\begin{proof}
	Path $ap$ will reach a layer-3 node after the first two edges (node $s$ does not appear twice), and then it will \textit{transfer} between layer-3 nodes or reach sink node $t$. As each MEC node in $ap$ does not appear twice, initially one layer-3 node has been visited, and the number of visited layer-3 nodes increases by one after each \textit{transfer}. With $M$ layer-3 nodes, the \textit{transfer} time does not exceed $M - 1$, and the edge number after each \textit{transfer} increases by 2. As a result, the edge number in the $ap$ does not exceed $2 \cdot M + 1$ (the \textit{transfer} does not increase the depth, and $ap$ need 3 edges to reach the sink node $t$).
\end{proof}

Using the above theorem, we propose the OKSP algorithm, which is equivalent to KSP algorithm, but has lower time complexity in our model. As the number of tiles $N$ in the scenario of this paper is much larger than the number of MEC servers $M$, the following analysis assumes that $N > M$.\\

\textbf{Optimization 1.} Use ordered arrays to maintain the shortest distance from the source node $s$ to every layer-3 node. \\

Assume in the augmenting pathfinding process, the path's endpoint is $s$, and it needs to go forward two edges ($s \rightarrow i \rightarrow j$) to reach MEC node $j$. There are $2 \cdot N$ possible paths, and if we check every tile node, the time complexity is $O(N)$. However, the path will eventually select the two edges with the highest profits (minimum sum of edge weights).

It doesn't matter which the first two edges are chosen. The path's endpoint always becomes node $j$, node $j$ becomes a visited MEC node, and edge set $\{e_{from, to}^k | from = s\ or\ to = j\}$ can no longer be chosen. In this case, it's better to pick the two edges that will generate the maximum profit.

Initially, we add all $2 \cdot N$ possible paths from node $s$ to each MEC node $j$ into the array $PathArray_j$ and sort their profits in descending order. During the pathfinding process, whenever a transfer from node $s$ to MEC node $j$ occurs, we simply take the path with the highest profit from the $PathArray_j$ and dynamically remove infeasible solutions at each iteration. 
For details, see \textit{Initialize-Possible-Paths} and \textit{Maintain-Path-Array} in Algorithm \ref{optimized-skp}. \\

\textbf{Optimization 2.} Use priority queues to maintain the shortest distance between the layer-3 nodes.\\

Assume in the augmenting pathfinding process, the path's endpoint is MEC node $j$, and ready to \textit{transfer} to MEC node $k$. Multiple transit nodes may be available for \textit{transit}, and the time complexity will be $O(N)$ if all tile nodes need to be checked. But it can be proved that the path will eventually choose the \textit{transit node} with the minimum \textit{transfer loss}. 

Because no matter which \textit{transit node} is selected, the path'endpoint always becomes MEC node $k$,  MEC node $k$ becomes a visited MEC node, and edge set $\{e_{from, to}^k | from = j\ or\ to = k\}$ can no longer be chosen. In this case, it's better to pick the \textit{transit node} that will cause the minimum \textit{transfer loss}.

The priority queue $LossQueue_{j,k}$ can be used to dynamically maintain all feasible \textit{transit nodes} of MEC server $j$ to transfer to MEC server $k$. For details, see \textit{Maintain-Loss-Queues} in Algorithm \ref{optimized-skp}.

After the above optimization, the tile nodes of layer-2 are hidden, and the original graph $G$ can be transformed into a new graph $G'$ with the number of nodes $M+2$ in Fig.~\ref{ksp2-fig}.

\begin{figure}
	\centering
	\includegraphics[width=1\textwidth]{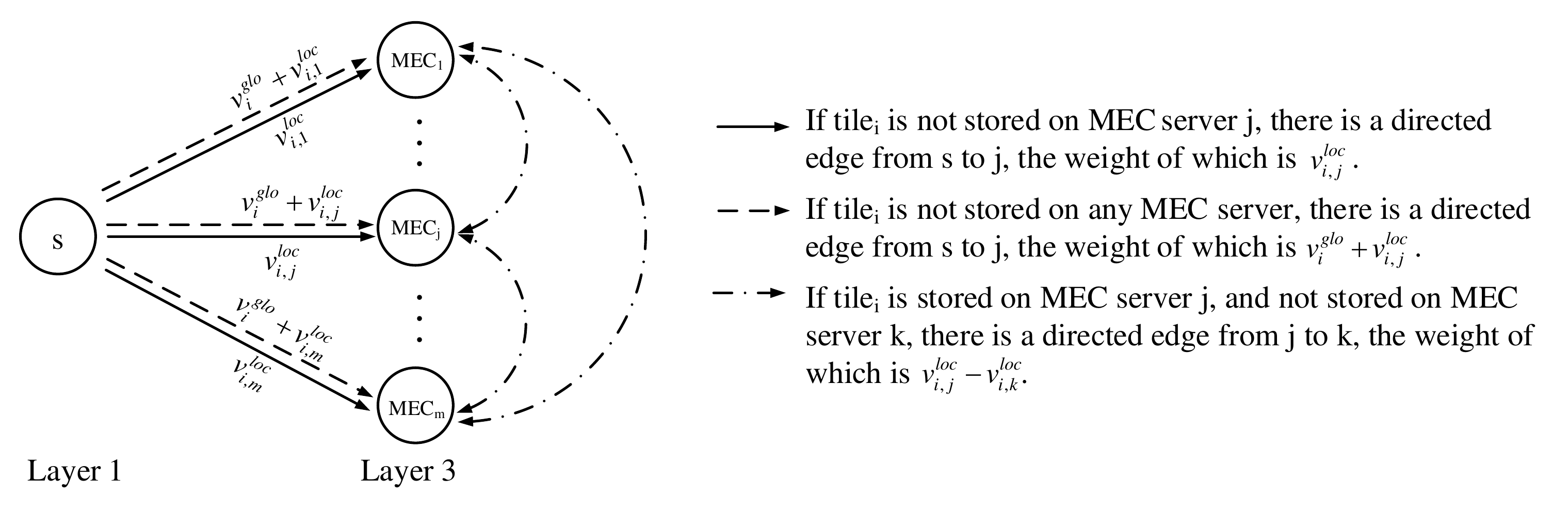}
	\caption{The illustration on optimizing the k-shortest problem. Note that the graph shows the possible cases of tile$_i$ profits in the collaborative domain.} \label{ksp2-fig}
\end{figure}

\begin{theorem}\label{Theorem 4}
	The total length of k shortest paths calculated by the OKSP algorithm is equal to the KSP algorithms.
\end{theorem}
\begin{proof}
	Assume that the augmenting path $ap_l$ obtained by $G'$ in $l^{th}$ iteration is ($s \rightarrow$ MEC$_1 \rightarrow  ... \rightarrow$ MEC$_j \rightarrow ... t$). From Optimization 1 and 2, we know that $s \rightarrow$ MEC$_1$ will greedily choose the path with the maximum profit, and MEC$_j \rightarrow$ MEC$_{j+1}$ will choose the transit node with the minimum \textit{transfer loss}, and the algorithm covers all available options when maintaining $PathArray$ and $LossQueue$, and the length of all shortest paths equals the length of all augmenting paths. Therefore, Theorem \ref{Theorem 4} holds.
\end{proof}

The following is an analysis of time complexity. The primary time consumption of the OKSP algorithm is as follows. (1) Calculating the augmenting path. There are K iterations, each time using the Dijkstra algorithm. The time complexity is $O(K \cdot M^2)$. (2) Maintaining MECPathArray. PathArray needs to sort the profits of all possible cache choices for each MEC node, and the time complexity is $O(M \cdot N \cdot \log N)$. (3) Maintaining MECLossQueue. Since the edge number of augmenting path obtained from each iteration will not exceed $O(M)$, the add operation will not exceed $O(K \cdot M^2)$ times. The time complexity of a priority queue add operation is $O(\log N)$; the remove operation times will not exceed the add operation and can be omitted.

Therefore, the total time complexity is $O((K\cdot M + N) \cdot M \cdot \log N)$. For random input, the edge number of augmenting path of each iteration is much smaller than $M$, and the run time will be faster than expected.

In the MEC collaboration domain, the collaboration system periodically runs the OKSP algorithm, selecting the least loaded MEC server in the collaboration domain to compute the caching scheme and distributing it to the other MEC servers. As soon as other MEC servers receive the cache scheme, they will request the tile files they do not have and need to cache.

\section{Experiment}
\label{sec:Experiment}
In order to evaluate the OKSP algorithm's low complexity and applicability to large-scale input data, we conducted numerical simulation experiments. In addition, we demonstrate that the OKSP algorithm can compute the caching scheme with the lowest average request latency in all cases by comparing it with other cooperative caching algorithms.

For the parameters, we set their values based on the fundamental physical situation and concerning similar studies.
We assume all MEC servers in the collaboration domain have the same storage space.
The VR device bandwidth is [50,100]Mbps, the MEC server bandwidth is 500 Mbps, and the cloud server bandwidth is 1 Gbps.

The tile file size is 10MB. Assume that a remediation request requests 10\% $\sim$ 20\% of a tile's size, and a prefetch request requests 50\% $\sim$ 100\% of a tile's size. 
As in most studies \cite{borst_distributed_2010,wang_zone-based_2019}, we assume the popularity of tiles follows a Zipf distribution and the shape parameter $\alpha$ of the Zipf distribution is 1.5. Generally, the greater the shape parameter $\alpha$, the more concentrated the popular content.
We assume the queuing delay $T_{q1} $ during communication is 1ms and $T_{q2}$ is 2ms. 
The packet transmission delay $T_{m\leftrightarrow dc}$ on the link between the VR device and the cloud server is [50,100]ms.

In the collaboration domain, a cache hit results in a lower request latency than a cloud response. 
The request latency reduction is called request latency optimization. And all request latency optimizations for the same tile equal the tile's caching profit. 
By reducing request latency, we aim to improve users' experience. Thus, we use the average request latency (ARL) optimization in the collaboration domain as a metric for evaluating the algorithm. We compare the OKSP algorithm with the following caching algorithms.
\begin{itemize}
	\item[-] \textit{Distributed}: Each MEC server in the collaboration domain has completely different tiles to achieve as many tiles as possible in the collaboration domain. The time complexity of the Distribution algorithm is $O(M \cdot N \cdot logN)$.
	
	\item[-] \textit{Self-Top}: Each MEC server in the collaboration domain caches the most popular tiles within its service area. The time complexity of the Self-Top algorithm is $O(M \cdot N \cdot logN)$.
	
	\item[-] \textit{MixCo}: In MixCo \cite{wang_zone-based_2019}, the storage space of each MEC server is divided into sub-zone dedicated storage space and zone-shared storage space. In the implementation of MixCo, after each dedicated space is allocated to a MEC server, an algorithm with time complexity $O(M \cdot N)$ is performed to calculate the average request latency in the current storage state of the collaboration domain. The worst time complexity of the MixCo algorithm is $O(L\cdot M^2 \cdot N)$.
\end{itemize} 
In the above time complexity representation, $L$ is the maximum number of tiles cached by a single MEC server, $M$ is the number of MEC servers, and $N$ is the total number of different tiles.

\subsection{Run time}
\begin{figure}
	\begin{minipage}[t]{0.3\linewidth}
		\centering
		\includegraphics[width=1.5in]{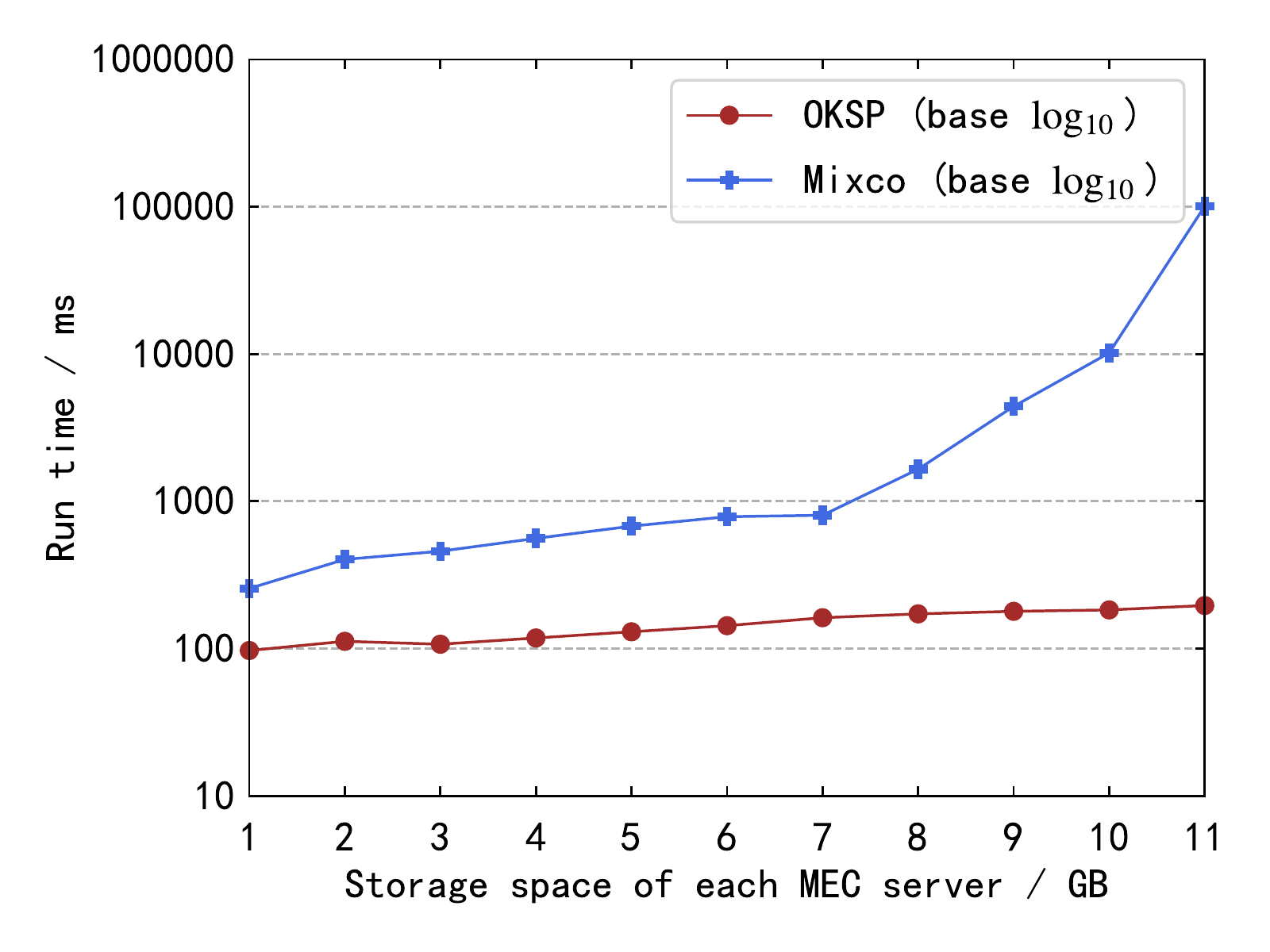}
		\caption{Run time comparison of the OKSP and Mixco algorithms for different storage space}
		\label{runtime-oksp-mixco}
	\end{minipage}
	\hspace{.15in}
	\begin{minipage}[t]{0.3\linewidth}
		\centering
		\includegraphics[width=1.5in]{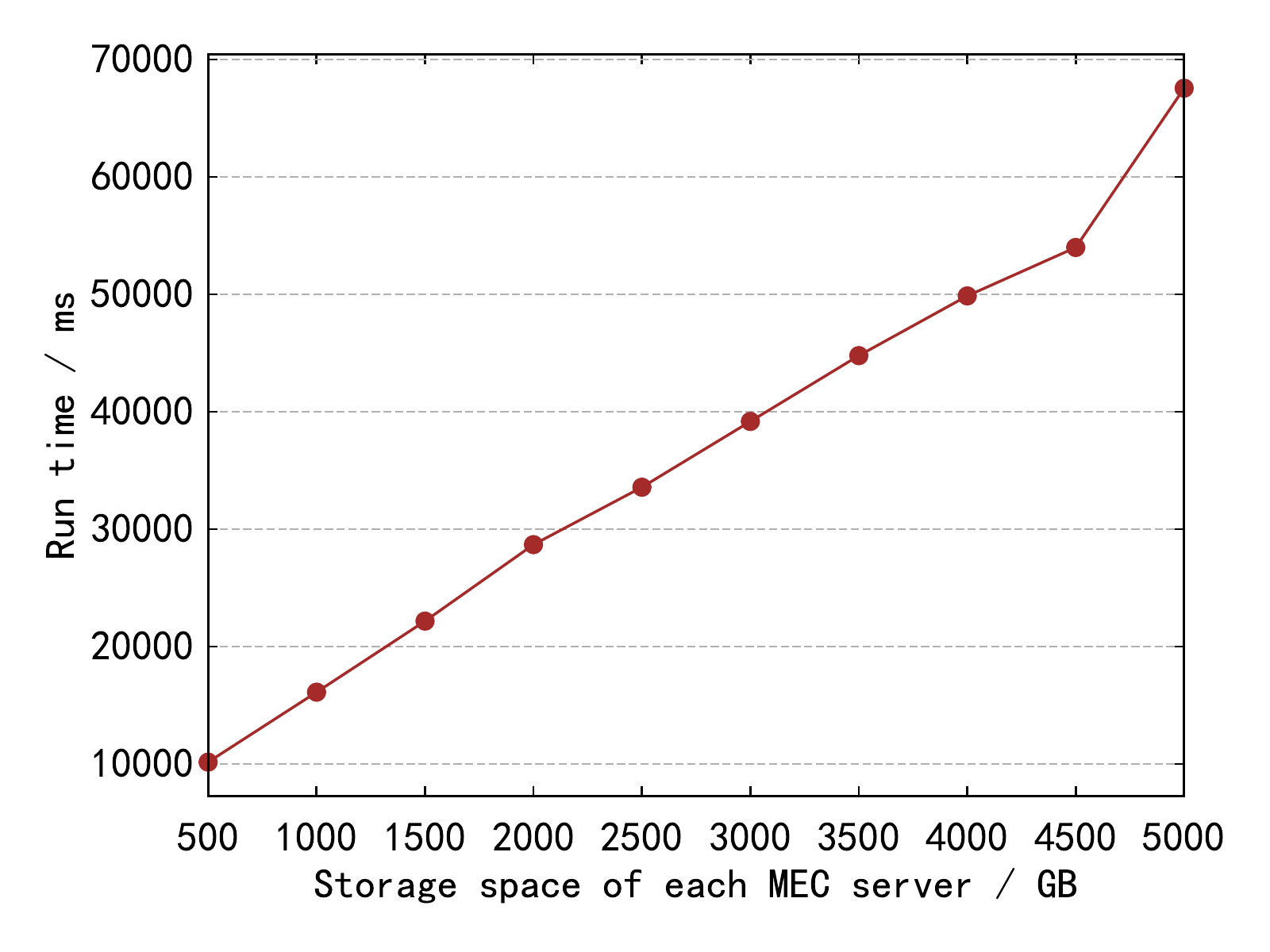}
		\caption{Run time of the OKSP algorithm for different storage space}
		\label{runtime-oksp-cm}
	\end{minipage}
\hspace{.15in}
\begin{minipage}[t]{0.3\linewidth}
	\centering
	\includegraphics[width=1.5in]{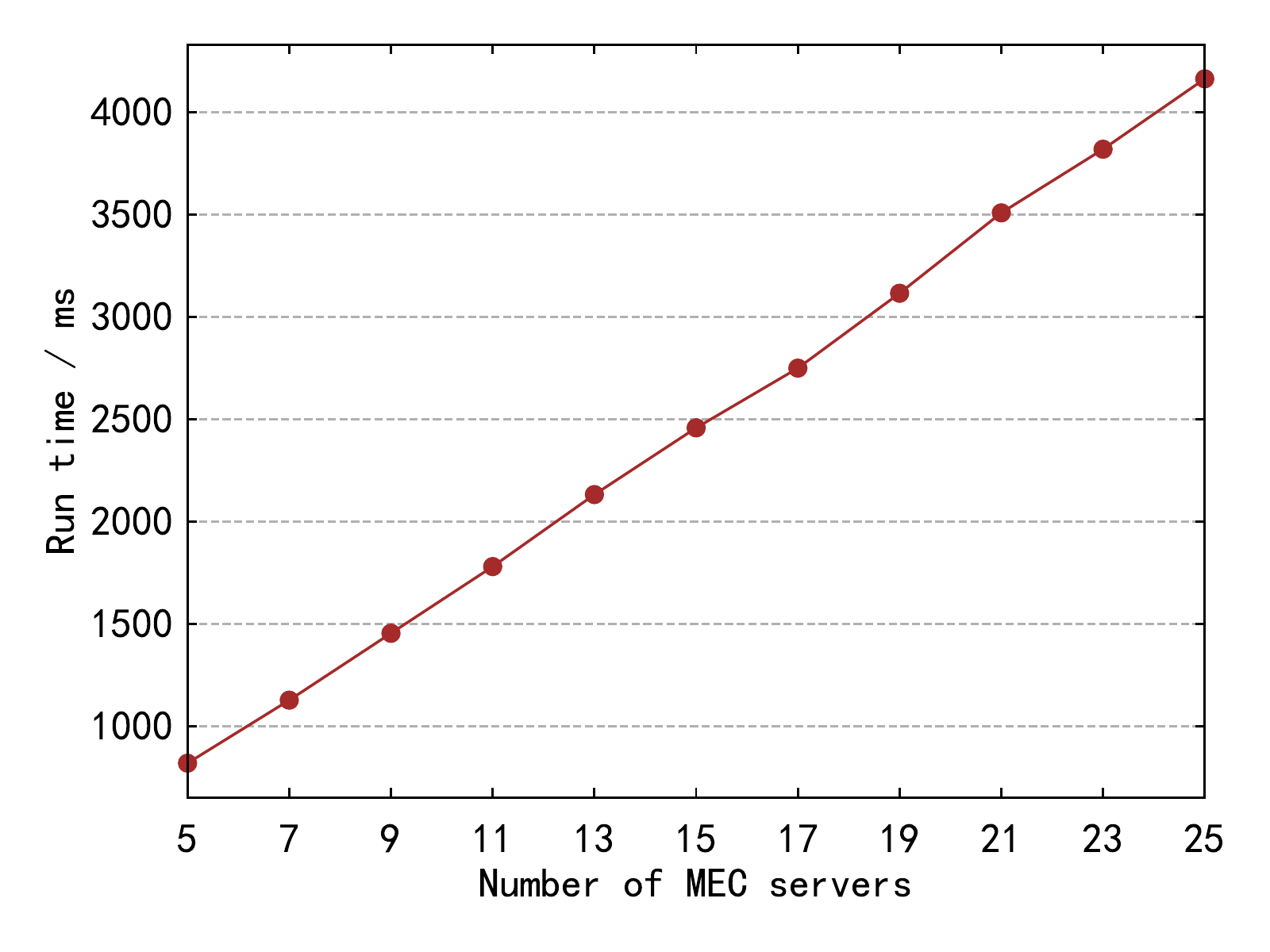}
	\caption{Run time of the OKSP algorithm for different number of MEC servers}
	\label{runtime-oksp-m}
\end{minipage}
\end{figure}
In this section, we will evaluate the speed of our OKSP algorithm. From the time complexity representation, the Self-Top and Distributed algorithms are faster than the OKSP algorithm. Thus, we do not take them into account. We perform a speed comparison experiment of the  OKSP and MixCo algorithms. 

This experiment fixes the number of tiles and MEC servers at 10,000 and 10, respectively, and adjusts the storage space of MEC servers.
The results are shown in Fig.~\ref{runtime-oksp-mixco}. Run time for the OKSP algorithm tends to increase linearly with increasing storage space, whereas the run time for the Mixco algorithm tends to grow exponentially. Due to more storage space available on MEC servers, MixCo generates multiple dedicated space partitioning iterations with time complexity $O(M \cdot N)$, which hinders algorithm performance. Therefore, applying the MixCo algorithm in a practical large-scale data environment is unrealistic.

To probe the speed limit of the OKSP algorithm, we further increase the data size. We fix the number of tiles and the number of MEC servers as 500,000 and 10, respectively, and increase the storage space of each MEC server. Fig.~\ref{runtime-oksp-cm} shows the results. The speed of OKSP increases linearly as storage space increases, which is consistent with the effect of storage space on its speed intuitively based on time complexity. 

In addition, we continued to experiment with the effect of MEC servers numbers variation on OKSP speed, and the results are shown in Fig.~\ref{runtime-oksp-m}. 
200,000 tiles and 2TB of storage space on each MEC server are fixed in this experiment. Due to the small order of magnitude of the MEC servers, the results indicate linear effects on OKSP speed. Generally, the OKSP algorithm is suitable for large-scale data problems.

\subsection{Performance Comparison}
In the following experiments, we specifically study the effects of the MEC storage space, the number of tiles and the shape parameter $\alpha$ of Zipf distribution on the performance of the above four caching algorithms. 
The MixCo algorithm does not apply to large-scale data, so we first test the performance of the four algorithms with smaller-scale data, then remove the MixCo algorithm and test the remaining three algorithms with large-scale data. The following results are the average of three simulation runs.
\begin{figure}
	\begin{minipage}[t]{0.3\linewidth}
		\centering
		\includegraphics[width=1.5in]{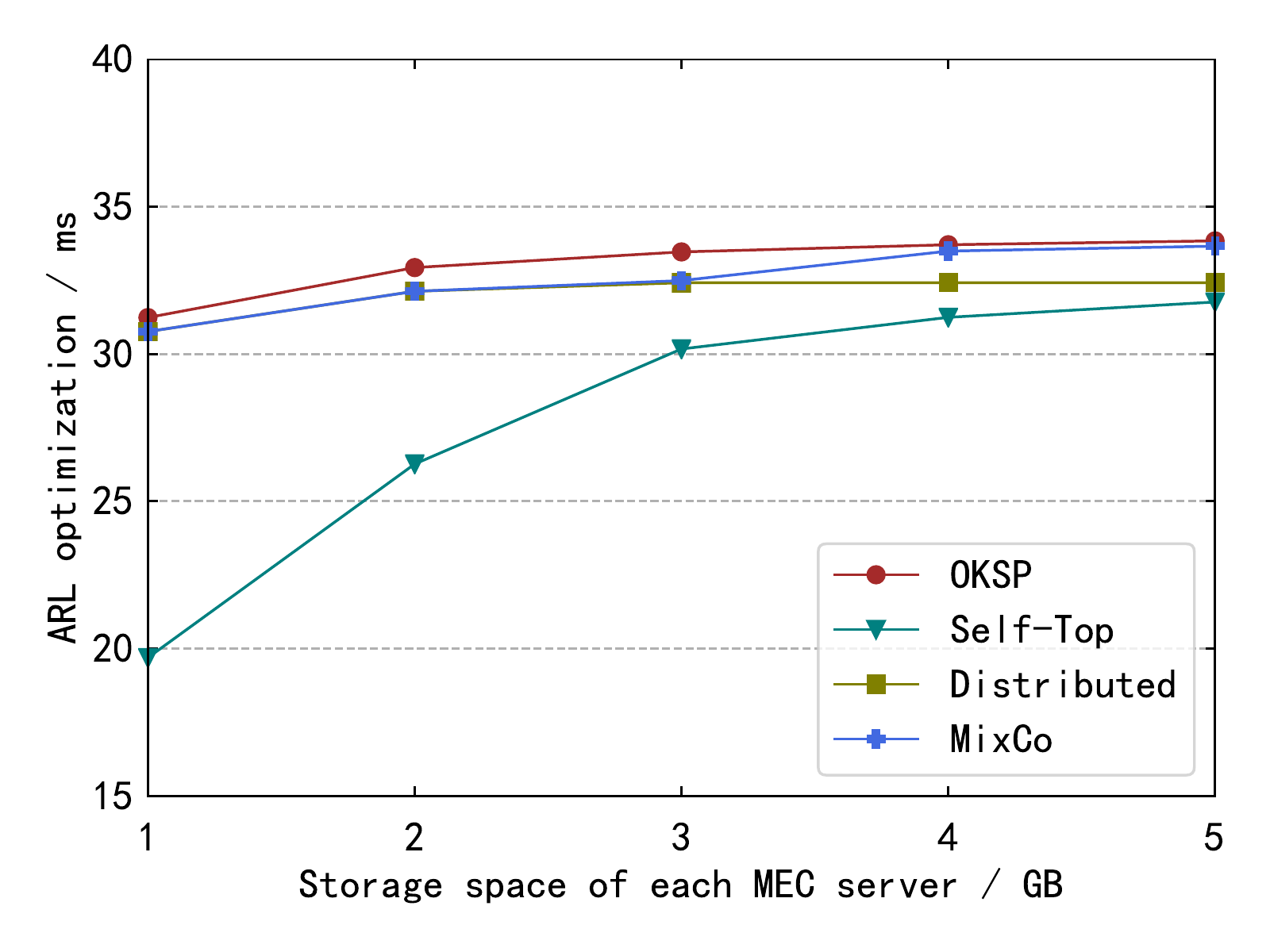}
		\caption{Performance comparison for different storage space}
		\label{pc-4-cm}
	\end{minipage}
	\hspace{.15in}
	\begin{minipage}[t]{0.3\linewidth}
		\centering
		\includegraphics[width=1.5in]{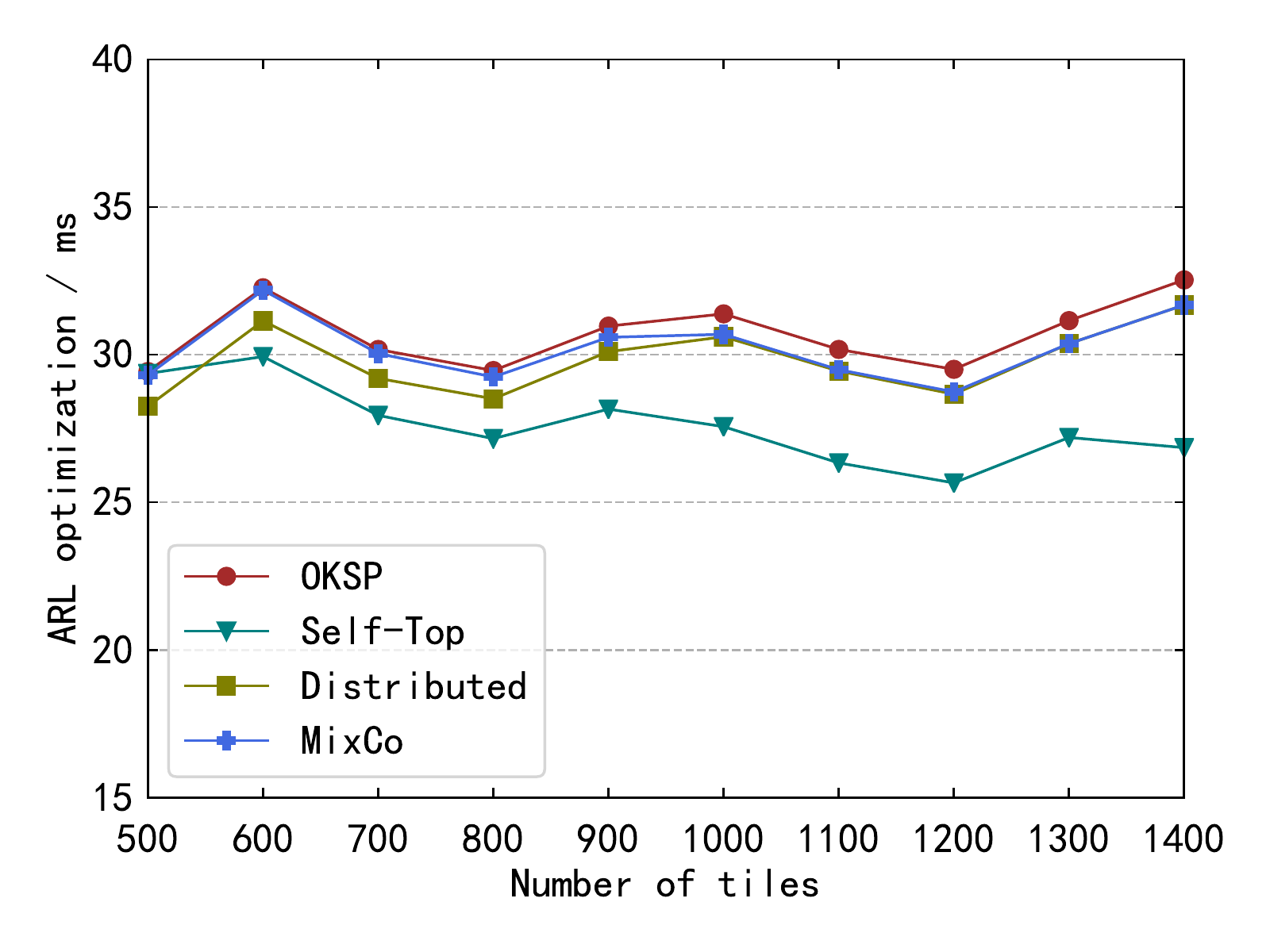}
		\caption{Performance comparison for different number of tiles}
		\label{pc-4-n}
	\end{minipage}
	\hspace{.15in}
	\begin{minipage}[t]{0.3\linewidth}
		\centering
		\includegraphics[width=1.5in]{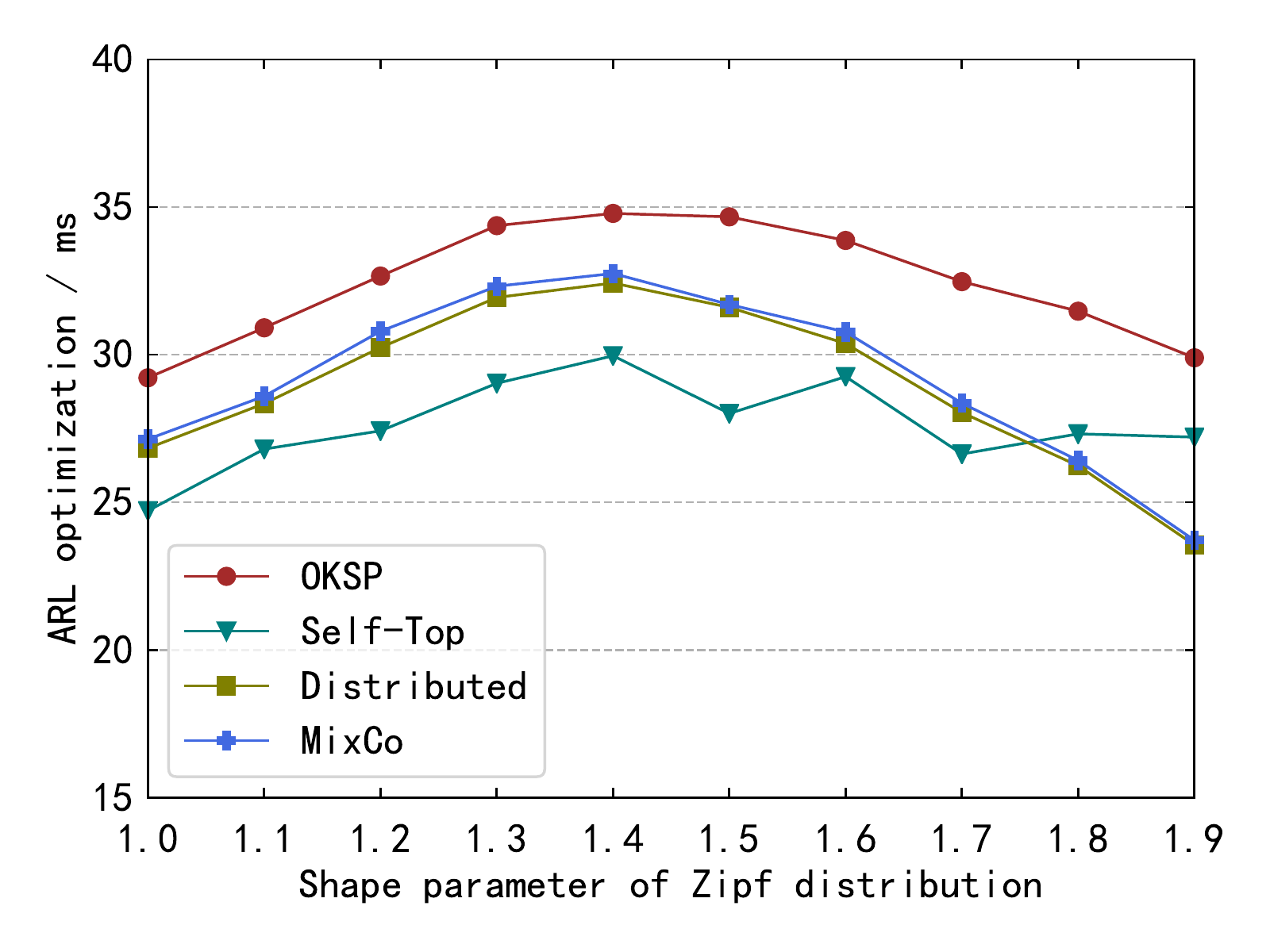}
		\caption{Performance comparison for different shape parameter $\alpha$}
		\label{pc-4-zpif}
	\end{minipage}
\end{figure}

First, we compare the ARL optimization when the storage space of the MEC server varies. We fix the number of MEC servers and the number of tiles as 5 and 1500, respectively. We increase the storage space of the MEC server from 1 GB to 5 GB. 
Fig.~\ref{pc-4-cm} shows the results. The ARL optimization of all caching algorithms increases as storage space increases. Because by increasing the storage space, more tiles can be cached in the MEC collaboration domain, preventing more requests from being sent to the cloud. 

Fig.~\ref{pc-4-n} depicts the ARL optimization for the different numbers of tiles. This experiment assumes the MEC server space is 2 GB and the MEC server number is 5. Due to the scattering of popular tiles and the different popular tiles on each MEC server, the Self-Top algorithm performs poorly.

Fig.~\ref{pc-4-zpif} shows the effect of the shape parameter $\alpha$ on ARL optimization. We fix the number of MEC servers is 5, MEC server storage space is 2 GB, and the number of tiles is 1000. In Fig.~\ref{pc-4-zpif}, we can observe that the ARL optimization of the four algorithms increases at the beginning as $\alpha$ increases. However, the ARL optimization of the four algorithms decreases as $\alpha$ improves further.
Because as $\alpha$ increases, the popular tiles become more concentrated. That is to say, a small number of tiles occupy most of the request traffic. Thus, the collaboration domain caching of these tiles can bring good profit.
As we continue to increase $\alpha$, despite the overall latency savings, the number of requests per tile increases, eventually showing a decrease in the ARL optimization.

\begin{figure}
	\begin{minipage}[t]{0.3\linewidth}
		\centering
		\includegraphics[width=1.5in]{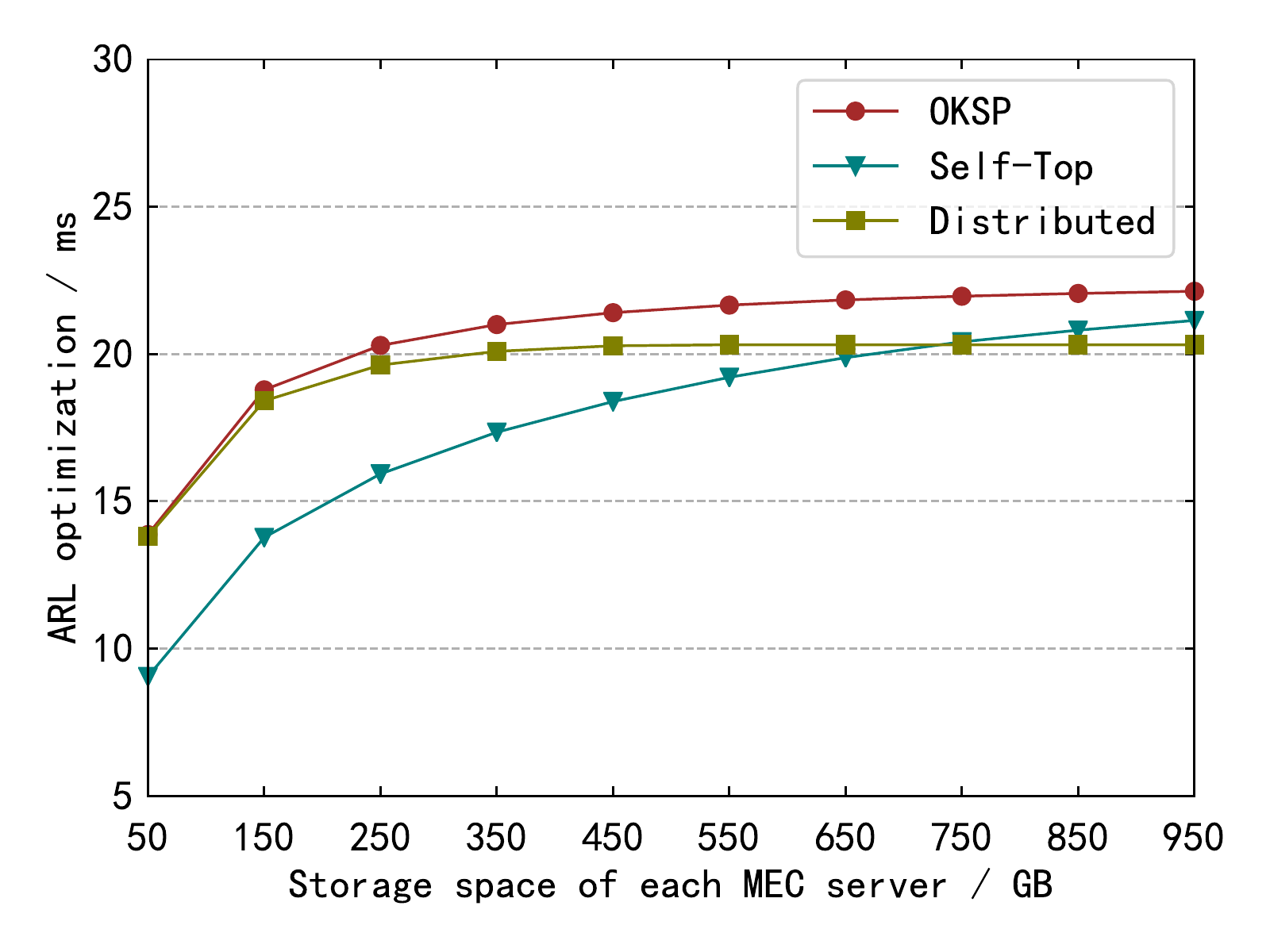}
		\caption{Performance comparison of the OKSP, Self-Top and Distributed algorithms for different MEC server storage space (popular tiles are similar on each MEC server)}
		\label{pc-3-cm-con}
	\end{minipage}
	\hspace{.15in}
	\begin{minipage}[t]{0.3\linewidth}
		\centering
		\includegraphics[width=1.5in]{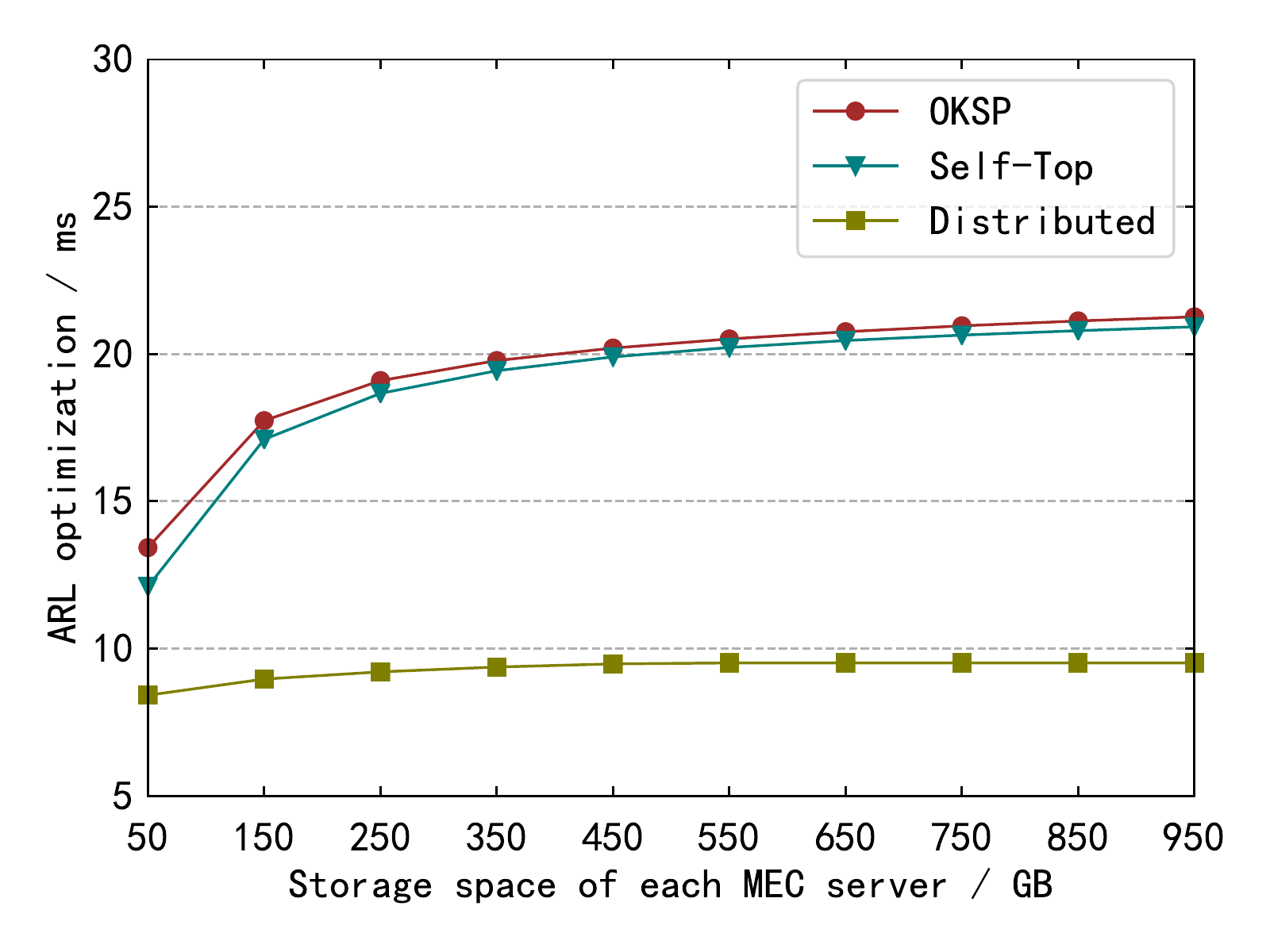}
		\caption{Performance comparison of the OKSP, Self-Top and Distributed algorithms for different MEC server storage space (popular tiles are different on each MEC server)}
		\label{pc-3-cm-sca}
	\end{minipage}
	\hspace{.15in}
	\begin{minipage}[t]{0.3\linewidth}
		\centering
		\includegraphics[width=1.5in]{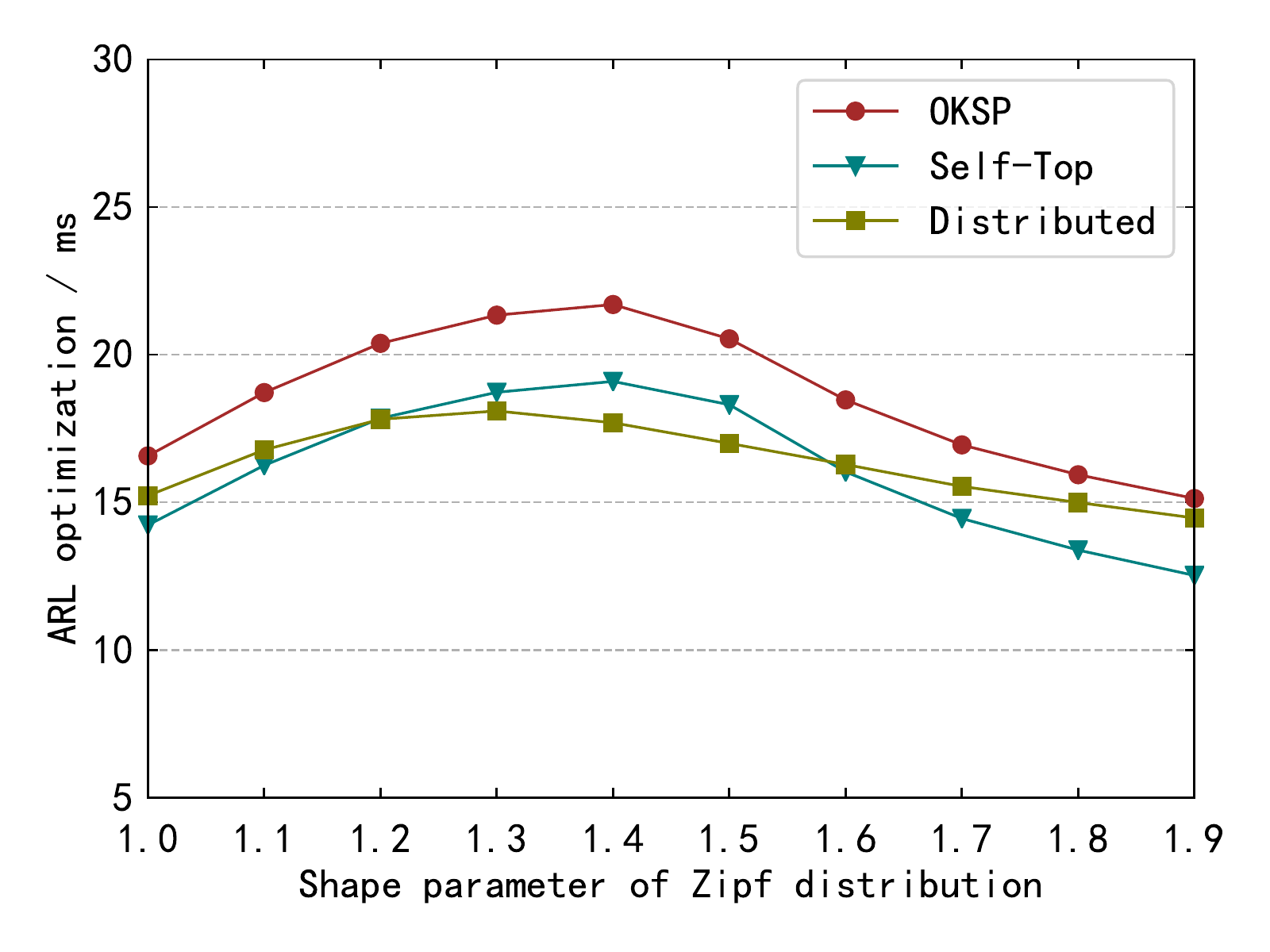}
		\caption{Performance comparison of the OKSP, Self-Top and Distributed algorithms for different shape parameter $\alpha$}
		\label{pc-3-zipf}
	\end{minipage}
\end{figure}

Finally, we test the ARL optimization of the remaining three algorithms on larger-scale data without Mixco. 
Fig.~\ref{pc-3-cm-con} and Fig.~\ref{pc-3-cm-sca} depict the ARL optimization under MEC server storage space variation. Assume there are 10 MEC servers and 500,000 tiles. In Fig.~\ref{pc-3-cm-con}, each MEC server's popularity is similar, while in Fig.~\ref{pc-3-cm-sca}, each MEC server's popularity is entirely random. 
Under different popularity distribution scenarios, the Self-Top and Distributed algorithms perform very differently. In contrast, the popularity distribution does not limit the performance of the OKSP algorithm, and it can keep steady regardless of whether the popular tiles are scattered or concentrated. 

ARL optimization under $\alpha$ variation is shown in Fig.~\ref{pc-3-zipf}. We assume the number of MEC servers is 10, MEC server storage space is 500 GB, and the number of tiles is 500,000. In Fig.~\ref{pc-3-zipf}, the results present a phenomenon similar to the experiment under small-scale data (in Fig.~\ref{pc-4-zpif}).
In general, the OKSP algorithm consistently performs better than the other algorithms.

\section{Conclusion}
\label{sec:Conclusion}
This paper provides a reliable transmission mechanism and an effective caching strategy under the current popular multi-MEC cooperative network architecture. 
We leverage MEC servers' computing and caching capabilities to improve users' viewing experience quality. 
For computing, we propose to let MEC servers handle prediction and remediation when predictions fail, while VR devices focus only on decoding and playing VR videos. For caching, VR videos are split into tile files, which improves cache hit rates.

To solve the order-of-magnitude huge tile caching problem, we first analyze the caching profit of tiles and then formulate the problem in graph theory as k-shortest paths problem. Then, we propose an OKSP algorithm to solve the reformulated problem quickly, which has an upper time complexity of $O((K \cdot M + N) \cdot M \cdot log N))$ and is suitable for shortest paths with restricted number of edges. Finally, we evaluate the proposed OKSP algorithm through run time and performance comparisons. Results show the OKSP algorithm is high-speed for solving large-scale data in our situation and has an average latency lower in any case than other caching schemes. 

As part of the model formulating process, we analyze the collaboration domain load in the average queueing delay. However, the actual situation is more complex. Introducing a reasonable model to describe the load on the MEC server and refining the model in this paper is the next step in our research plan.

\appendix
\section{Appendix - KSP Algorithm}
In this appendix, Algorithm \ref{orignal-skp} gives a pseudocode summary of the original k-shortest paths algorithm\cite{suurballe_disjoint_1974}, Algorithm \ref{optimized-skp} gives a pseudocode of the optimized k-shortest paths algorithm.
In algorithm \ref{optimized-skp}, there will be edges with negative weight in the graph after many iterations, we use the tips in the Johnson's algorithm \cite{johnson1977efficient} to re-weight all edges and make them all positive, then apply the Dijkstra algorithm to compute the shortest path trees at each iteration. We refer the interested reader to \cite{suurballe_disjoint_1974} for further details.

\begin{algorithm}
	\label{orignal-skp} 
	\Fn{Reverse($G, p$)} {
		\tcc{Reverse the direction of all edges in $p$, \;and negating their weights.}
		\For{$e_{i,j}^t \in p$} {
			Add $e_{j,i}^t$ in $G$\;
			$c(e_{j,i}^t) = -c(e_{i,j}^t)$
			Remove $e_{i,j}^t$ from $G$\;
		}
		\textbf{return} $G$
	}
	
	\caption{Original-K-Shortest-Paths}
	\Fn{Original-K-Shortest-Paths} {
		Construct the initial graph G described in 4.2\;
		\For{$l \gets 1$ to $K$}{
			$AugmentingPath_l \gets$ \textit{Shortest-Path-Algorithm}$(G, s, t)$\;
			\If{$AugmentingPath_l$ not exists} {
				Break
			}
			$G \gets$ Reverse($G, AugmentingPath_l$)
		}
		$P$ $\gets$ find the shortest paths in $G$ based on the negative edges\;
		\textbf{return} $P$
	}
	
\end{algorithm}

\begin{algorithm}
	\caption{Optimized-K-Shortest-Path Part-1}
	\label{optimized-skp}
	\LinesNumbered
	\Fn{Initialize-Possible-Paths}{
		\For{$j \gets 1$ to $M$}{
			\For{$i\gets 1$ to $N$}{
				\tcc{Two kinds of profits of tile$_i$ store in MEC server $j$.}
				Add $(v_{i,j}^{loc} + v_i^{glo}, i)$ to $PathArray_j$\;
				Add $(v_{i,j}^{loc}, i)$ to $PathArray_j$\;
			}
			Sort $PathArray_j$ in descending order
		}
	}
	
	\Fn{Maintain-Path-Array}{
		\For{$j \gets 1$ to $M$} {
			\tcc{Remove infeasible paths from $PathArray$.}
			\While{$PathArray_j$ is not empty} {
				$(profit, i) \gets$ the maximum profit element in $PathArray_ j$\;
				\If{tile$_i$ is stored in the collaboration domain} {
					\tcc{Check if the global profit of tile$_i$ has been accrued.}
					\If{$profit > v_{i,j}^{loc}$  \textbf{or} tile$_i$ is stored in MEC server $j$} {
						Remove $(profit, i)$ from $PathArray_j$\;
						Continue
					}
				}
				Break\;
			}
		}
	}
	\Fn{Setup-Graph(G')} {
	\tcc{Because the shortest path algorithm, negating the weights of all edges}
	Initialize All nodes in G' are unreachable to each other
	\For{$j \gets 1$ to $M$} {
		\If{$PathArray_j$ is not empty}{
			$(profit, i) \gets$ the maximum profit element in $PathArray_ j$\;
			$c(e_{s, j}) \gets -profit$
		}
		\If{MEC server j has storage space left} {
			$c(e_{j, t}) \gets 0$
		}
		\For{$k \gets 1$ to $k$ \textbf{and} $k \ne j$} {
			\If{$LossQueue_{j,k}$ is not empty} {
				$(loss, i) \gets$ the minimum loss element in $LossQueue_j,k$\;
				$c(e_{j, k}) \gets -loss$
			}
		}
	}
}
\end{algorithm}

\begin{algorithm}
	\LinesNumbered
	\setcounter{AlgoLine}{26}
	\Fn{Maintain-Loss-Queues(spath, CacheTable)} {
		\tcc{The edges of $spath$ are denoted as $(j \rightarrow k, i)$, indicating that MEC node $j$ is transferred to node $k$ via tile$_i$; In particular, the first edge of spath $(s \rightarrow sj, si)$ indicates that the MEC server $sj$ newly store the tile $si$, which is not transferred from another MEC; the last edge $(tj \rightarrow t, null)$ does not contain the transferred tile; and the remaining edges are considered \textit{transferred edges}.}
		Mark MEC server $sj$ is caching tile$si$ in $CacheTable$\;
		The number of copies of tile$_{si}$ cached in the collaboration domain += 1\;
		The remaining storage space of MEC server $tj$ -= 1\;
		\tcc{Adding the new transfer nodes.}
		\For{$(remove \rightarrow add, i) \in$ transferred edges}{
			Mark MEC server $add$ is caching tile$_i$ in $CacheTable$\;
			Mark MEC server $remove$ no longer cache tile$_i$ in $CacheTable$\;
			\For{$j \gets 1$ to $M$} {
				\If{$j \ne add$ \textbf{and} MEC j does not cache tile$_i$ } {
					Add $(v_{i,add}^{loc} - v_{i,k}^{loc}, j)$ to $LossQueue_{add, j}$
				}
				\If{$j \ne remove$ \textbf{and} MEC j is caching tile$_i$ } {
					Add $(v_{i,j}^{loc} - v_{i,remove}^{loc}, j)$ to $LossQueue_{j, remove}$
				}
			}
		}
		\tcc{Remove infeasible transfer nodes.}
		\For{$j \gets 1$ to $M$} {
			\For{$k \gets 1$ to $k$ \textbf{and} $k \ne j$} {
				\While{$LossQueue_{j,k}$ is not empty} {
					$(loss, i) \gets$ the minimum loss element in $LossQueue_j,k$\;
					\If{MEC server $j$ does not cache tile$_i$ \textbf{or} MEC server $k$ is caching tile$_i$} {
						Remove $(loss, i)$ from $LossQueue_j,k$\;
						Continue		
					}
					Break
				}
			}
		}
	}
	
	\Fn{Optimized-K-Shortest-Paths} {
		CacheTable $\gets$ a hash table for recording cache status\;
		Construct $G'$ containing $M$ MEC nodes, $s$ node, and $t$ node\; 
		Initialize-Possible-Paths()\;
		Setup-Graph($G'$)\;
		$p_1 \gets$ \textit{Flody-Algorithm}($G', s, t$)\;
		Re-Weight-All-Edges(the shortest path trees node $s$)
		\For{$l \gets 1$ to $K$}{
			Setup-Graph($G'$)\;
			$p_l \gets$ \textit{Dijkstra-Algorithm}($G', s, t$)\;
			\If{$p_l$ not exists} {
				Break
			}
			Re-Weight-All-Edges(the shortest path trees node $s$)
			MaintainLossQueues($p_l$, CacheTable)\;
		}
		\textbf{return} CacheTable
	}
\end{algorithm}

%
%
%

\clearpage
\bibliographystyle{splncs04}
\bibliography{samplepaper}

\end{document}